\documentclass{article}%
\usepackage{amsfonts}
\usepackage{amsmath}%
\usepackage{amssymb}%
\usepackage{graphicx}
\usepackage[english]{babel}
\providecommand{\U}[1]{\protect\rule{.1in}{.1in}}
\newtheorem{theorem}{Theorem}

\newtheorem{corollary}[theorem]{Corollary}

\newtheorem{example}[theorem]{Example}

\newtheorem{lemma}[theorem]{Lemma}

\newtheorem{proposition}[theorem]{Proposition}

\newenvironment{proof}[1][Proof]{\noindent\textbf{#1.} }{\hfill $\blacksquare$}

\def\equ #1{\begin{equation}#1\end{equation}}

\def\sbra#1{\left(#1\right)}

\def\lbra #1{\left\{#1\right\}}
\def\diag #1{\text{diag}#1}

\newcommand{\keff}{\kappa_{\mathrm{eff}}}

\begin{document}

\title{Positivity of a Hadamard Product}
\author{Roger A. Horn\thanks{Tampa, FL, USA, email: \texttt{rhorn@math.utah.edu}}, Shengxuan
Luo\thanks{School of Mathematics
	and Statistics, Xi'an Jiaotong University, Xi'an 710049, China, emails: \texttt{2592425519@stu.xjtu.edu.cn; xhwice@stu.xjtu.edu.cn; yangzai@xjtu.edu.cn}}, Hongwei Xu$^\dagger$, and Zai Yang$^\dagger$\thanks{Corresponding author}}
\maketitle

\begin{abstract}
A notable difference between the ordinary and Hadamard products is that the
Hadamard product of two singular positive semidefinite matrices can be
nonsingular.
We present an eigenvalue lower bound for a Hadamard product that depends on
a leading spectral projection of freely selected rank for one factor, and the smallest eigenvalues of certain principal submatrices of the
other factor. Maximizing over the selected rank gives a stronger estimate;
an effective-condition-number bound follows as a simpler consequence.
We establish sharpness for prescribed ranks and effective condition numbers
and examine limiting examples that delineate the bounds' usefulness.
We give numerical examples and discuss applications in array signal processing
and matrix time series analysis.
\end{abstract}

\section{Introduction}\label{sec:intro}
The Schur Product Theorem says that the Hadamard product $A\circ B=\left[
a_{ij}b_{ij}\right]  $ of positive semidefinite $n\times n$ complex matrices
$A=\left[  a_{ij}\right]  $ and $B=\left[  b_{ij}\right]  $ is positive
semidefinite; see \cite{schur1911bemerkungen} or \cite[Theorem 7.5.3]{horn2012matrix}, and \cite{yangliu,MondPecaric} for its various generalizations. If $A$ is
positive definite and $B$ has no zero diagonal entries, then $A\circ B$ is
positive definite even if $B$ is singular. In fact, there is a sharp lower bound% 
\begin{equation}
\lambda_{\min}(A\circ B)\geq\lambda_{\min}(A)\min_{1\leq i\leq n}%
b_{ii}\label{classical lower bound}%
\end{equation}
on the smallest eigenvalue of $A\circ B$ \cite{schur1911bemerkungen}. But one can say more.

If $n>2,$ then $A\circ B$ can be nonsingular even if both $A$ and $B$ are
singular. For example, consider $n=3,$%
\begin{equation}
A=\left[
\begin{array}
[c]{ccc}%
2 & 1 & 1\\
1 & 1 & 0\\
1 & 0 & 1
\end{array}
\right]  ,\quad B=\left[
\begin{array}
[c]{ccc}%
2 & 1 & 1\\
1 & 1 & 1\\
1 & 1 & 1
\end{array}
\right]  ,\quad\text{and}\quad A\circ B=\left[
\begin{array}
[c]{ccc}%
4 & 1 & 1\\
1 & 1 & 0\\
1 & 0 & 1
\end{array}
\right]  ;\label{StyanExample}%
\end{equation}
see \cite{styan1973hadamard}. The matrices $A$ and $B$ are positive semidefinite and $\det
A=\det B=0,$ but $\det\left(  A\circ B\right)  =2,$ so $A\circ B$ is positive
definite. This is not an isolated curiosity. Nonsingularity of $A\circ B$ can
be predicted from observing that every principal minor of $A$ of size 2 is
positive (its Kruskal rank $k_{A}=2$), the rank of $B$ is $2$ ($r_B=2$)$,$
and $k_{A}\geq n-r_B+1.$ In \cite{yang2019source} and \cite[Theorem 4]{horn2020rank} we showed
that the latter condition, if $B$ has no zero main diagonal entry,
is sufficient for nonsingularity of the Hadamard product of positive semidefinite
matrices of any size. In this paper, we sharpen this qualitative result by
establishing an explicit positive lower bound for the smallest eigenvalue of
$A\circ B$ when $k_{A}\geq n-r_B+1.$

Our main bound uses a leading spectral projection of one factor, whose rank can be selected
freely. Maximizing over this rank strengthens the estimate, and an
effective-condition-number bound follows as a simpler consequence.

The remainder of the paper is organized as follows. Section~\ref{sec:notation} introduces notation and some basic facts, and states the main result.
Section~\ref{sec:proof} proves it using Hadamard products with orthogonal
projections. Section~\ref{sec:sharp} treats sharpness and shifted positivity,
and Section~\ref{sec:examples} gives quantitative comparisons.
Sections~\ref{sec:array} and \ref{sec:time} discuss applications to array
signal processing and matrix time series, respectively.
Section~\ref{sec:conclusion} concludes the paper.

\section{Notation, basic facts, and main result}\label{sec:notation}

All of our matrices have real or complex entries unless otherwise stated. The set of $n\times n$
matrices is denoted by $M_{n}.$ The conjugate transpose of a matrix $A$ is
$A^{\ast},$ and $A$ is \emph{Hermitian} if $A=A^{\ast}.$ A Hermitian $A\in
M_{n}$ is \emph{positive semidefinite }(denoted by $A\succeq0)$ if $z^{\ast
}Az\geq0$ for all nonzero $z\in\mathbb{C}^{n};$ it is \emph{positive definite}
(denoted by $A\succ0)$ if $z^{\ast}Az>0$ for all nonzero $z\in\mathbb{C}^{n}.$
Every principal submatrix of a positive semidefinite matrix is positive
semidefinite. For Hermitian $A,B\in M_{n},$ we write $A\succeq B$ if
$A-B\succeq0.$ Vectors (matrices with one column) are in boldface; the
Euclidean norm of a vector $\mathbf{x}$ is denoted by $\left\Vert
\mathbf{x}\right\Vert =\left(  \mathbf{x}^{\ast}\mathbf{x}\right)  ^{1/2}.$

The eigenvalues of a Hermitian $A\in M_{n}$ are real; they are nonnegative if
$A\succeq0;$ and they are positive if $A\succ0.$ The \emph{rank} of a nonzero
matrix $A\in M_{n}$ (denoted by $r_{A})$ is the largest positive integer $q$
such that some list of $q$ distinct columns of $A$ is linearly independent; if
$A=0,$ then $r_{A}=0.$ If $A$ is Hermitian, then $r_{A}$ is the number of its
nonzero eigenvalues. We enumerate the eigenvalues of a Hermitian $A\in M_{n}$
in algebraically non-increasing order: $\lambda_{1}\geq\lambda_{2}\geq
\cdots\geq\lambda_{n}=\lambda_{\min}.$ A matrix $P\in M_{n}$ is an
\emph{orthogonal projection} if it is Hermitian and idempotent, that is,
$P=P^{\ast}$ and $P^{2}=P.$

The \emph{Kruskal rank }of a matrix $A$ with no zero column (denoted by
$k_{A})$ is the largest positive integer $q$ such that every list of $q$
distinct columns of $A$ is linearly independent; if $A$ has a zero column,
then $k_{A}=0.$ We always have $0\leq k_{A}\leq r_{A}.$ If $A\in M_{n}$ is
positive semidefinite, then $k_{A}\geq q\geq1$ if and only if every $q\times
q$ principal submatrix of $A$ is positive definite, that is, if and only if
every principal minor of $A$ of size at most $q$ is positive; see
\cite[Corollary 3]{horn2020rank}. A positive semidefinite matrix has no zero column if
and only if it has no zero main diagonal entry.

A sufficient condition for positive definiteness of the Hadamard product of
positive semidefinite matrices involves the Kruskal rank of one factor and the
rank of the other; see \cite[Theorem 3]{yang2019source} and \cite[Theorem 4]{horn2020rank}.

\begin{theorem}
\label{YangTheorem}\label{thm:prior}
Let $A,B\in M_{n}$ be positive semidefinite. If $B$ has no
zero main diagonal entry and $k_A + r_B \geq n+1$, then $A\circ B$ is positive definite.
\end{theorem}

For a set \(S\subseteq\{1,\ldots,n\}\), write \(A[S]\) for
the principal submatrix of \(A\) with rows and columns indexed by \(S\), and define
\[
\mu_j(A)=\min_{|S|=j}\lambda_{\min}(A[S]),\qquad j\in\{1,\ldots,n\}.
\]
For a nonzero \(B\succeq0\), its effective condition number is
\(\keff(B)=\lambda_1(B)/\lambda_{r_B}(B)\).
The following theorem is our main quantitative refinement of Theorem~\ref{YangTheorem}.
For a nonzero \(B\succeq0\), let \(P_r\) project onto the span of
\(r\) orthonormal eigenvectors associated with its \(r\) largest eigenvalues,
\(r\in\{1,\ldots,r_B\}\). In particular,
\(P_{r_B}\) is the orthogonal projection onto the range of \(B\). At an eigenvalue tie,
an intermediate \(P_r\) need not be unique; every such choice
is allowed in the following result.

\begin{theorem}\label{Big}\label{thm:main}
Let \(A,B\in M_n\) be positive semidefinite and suppose that \(B\ne0\). For every
\(r\in\{1,\ldots,r_B\}\), we have
\begin{equation}\label{eq:main}
A\circ B\succeq
\mu_{n-r+1}(A)\,\lambda_r(B)(I\circ P_r).
\end{equation}
Consequently,
\begin{equation}\label{eq:optimized}
\lambda_{\min}(A\circ B)\ge
\max_{1\le r\le r_B}
\left\{\mu_{n-r+1}(A)\,
\lambda_r(B)\min_{1\le i\le n}(P_r)_{ii}\right\}.
\end{equation}
\end{theorem}

We refer to \eqref{eq:optimized} as the {\em optimized bound}.
Taking \(r=r_B\) in \eqref{eq:main} gives the {\em full-range matrix bound}
and, consequently, the {\em full-range scalar bound}
\begin{equation}\label{eq:full-range}
\lambda_{\min}(A\circ B)\ge
\mu_{n-r_B+1}(A)\lambda_{r_B}(B)
\min_{1\le i\le n}(P_{r_B})_{ii}.
\end{equation}

\begin{corollary}\label{cor:condition}
Under the hypotheses of Theorem~\ref{Big},
\begin{equation}\label{eq:condition-matrix}
A\circ B\succeq
\frac{\mu_{n-r_B+1}(A)}{\keff(B)}(I\circ B),
\end{equation}
and
\begin{equation}\label{eq:condition-scalar}
\lambda_{\min}(A\circ B)\ge
\frac{\mu_{n-r_B+1}(A)}{\keff(B)}\min_i b_{ii}.
\end{equation}
\end{corollary}

We call \eqref{eq:condition-scalar} the {\em effective-condition-number bound}.

Theorem~\ref{Big} also applies when one or both factors are nonsingular.
If \(B\succ0\), then \(P_n=I\), and the \(r=n\) term in
\eqref{eq:optimized} is
\[
\mu_1(A)\lambda_n(B)
=\min_i a_{ii}\,\lambda_{\min}(B).
\]
Interchanging \(A\) and \(B\) recovers \eqref{classical lower bound}
when \(A\succ0\). Thus, taking the better of the two orientations of
\eqref{eq:optimized} recovers both classical Schur bounds, while
intermediate ranks can improve them. See Example~\ref{ex:positive-definite}.

\section{Proof of the main result}\label{sec:proof}
Our proof relies on Hadamard products in which one factor is an orthogonal
projection. We first need to understand the structure of an orthogonal
projection that is presented as a bordered matrix.
\begin{lemma}
\label{Lemma}\label{lem:border}Let $n\geq2$ and let $P=[p_{ij}]\in M_{n}$ be a rank-$r$
orthogonal projection, partitioned as%
\[
P=\left[
\begin{array}
[c]{cc}%
P_{1} & \mathbf{x}\\
\mathbf{x}^{\ast} & p
\end{array}
\right]  ,
\]
in which $P_{1}\in M_{n-1}$ and $\mathbf{x}\in\mathbb{C}^{n-1}.$ Then
\newline
(a) $\left\Vert \mathbf{x}\right\Vert ^{2}=p(1-p)$ and $0\leq p\leq1.$
\newline
(b) If $p=0$ or $p=1,$ then $\mathbf{x}=\mathbf{0}$ and $P_{1}$ is a
rank-$(r-p)$ orthogonal projection.
\newline
(c) If $0<p<1,$ then $1\leq r\leq n-1$ and%
\[
P=\left[
\begin{array}
[c]{cc}%
Q+\frac{1}{p}\mathbf{xx}^{\ast} & \mathbf{x}\\
\mathbf{x}^{\ast} & p
\end{array}
\right]  ,
\]
in which $Q=P_{1}-\frac{1}{p}\mathbf{xx}^{\ast}$ is a rank-$(r-1)$ orthogonal
projection and $Q\mathbf{x}=\mathbf{0}$.
\newline
(d) If $0<p<1$, then the matrix $R=Q+\frac{1}{\left\Vert \mathbf{x}\right\Vert ^{2}}%
\mathbf{xx}^{\ast}=Q+\frac{1}{p(1-p)}\mathbf{xx}^{\ast}$ is a rank-$r$
orthogonal projection.
\end{lemma}

\begin{proof}
The matrix $P$ is idempotent, so%
\[
P=\left[
\begin{array}
[c]{cc}%
P_{1} & \mathbf{x}\\
\mathbf{x}^{\ast} & p
\end{array}
\right]  =\allowbreak\left[
\begin{array}
[c]{cc}%
P_{1}^{2}+\mathbf{xx}^{\ast} & P_{1}\mathbf{x}+p\mathbf{x}\\
\mathbf{x}^{\ast}P_{1}+p\mathbf{x}^{\ast} & p^{2}+\mathbf{x}^{\ast}\mathbf{x}%
\end{array}
\right]  =P^{2}.
\]
Comparison of block entries shows that%
\[
P_{1}^{2}=P_{1}-\mathbf{xx}^{\ast},\quad P_{1}\mathbf{x}=(1-p)\mathbf{x}%
,\text{ \quad and\quad}\left\Vert \mathbf{x}\right\Vert ^{2}=p(1-p).
\]
Since $\left\Vert \mathbf{x}\right\Vert ^{2}=p(1-p)\geq0,$ it follows that
$0\leq p\leq1.$ Moreover, if $p=0$ or $p=1,$ then $\left\Vert \mathbf{x}%
\right\Vert =0,$ the direct summands of $P=P_{1}\oplus\lbrack p]$ are
orthogonal projections, and $\operatorname*{rank}P_{1}=\operatorname*{tr}%
P_{1}=\operatorname*{tr}P-p=r-p$.
\newline
\indent
Now suppose that $0<p<1.$ The matrix $Q$ is Hermitian,%
\[
Q\mathbf{x}=P_{1}\mathbf{x}-\frac{1}{p}\mathbf{x}\left(  \mathbf{x}^{\ast
}\mathbf{x}\right)  =(1-p)\mathbf{x}-\frac{p(1-p)}{p}\mathbf{x}=\mathbf{0},
\]
and%
\begin{align*}
Q^{2}  & =\left(  P_{1}-\frac{1}{p}\mathbf{xx}^{\ast}\right)  ^{2}\\
& =P_{1}^{2}-\frac{1}{p}\left(  P_{1}\mathbf{x}\right)  \mathbf{x}^{\ast
}-\frac{1}{p}\mathbf{x}\left(  P_{1}\mathbf{x}\right)  ^{\ast}+\frac{1}{p^{2}%
}\mathbf{x}\left(  \mathbf{x}^{\ast}\mathbf{x}\right)  \mathbf{x}^{\ast}\\
& =\left(  P_{1}-\mathbf{xx}^{\ast}\right)  -\frac{2(1-p)}{p}\mathbf{xx}%
^{\ast}+\frac{p(1-p)}{p^{2}}\mathbf{xx}^{\ast}\\
& =P_{1}-\frac{1}{p}\mathbf{xx}^{\ast}=Q,
\end{align*}
so $Q$ is an orthogonal projection. Moreover, $\operatorname*{rank}%
Q=\operatorname*{tr}Q=\operatorname*{tr}\left(  P_{1}-\frac{1}{p}%
\mathbf{xx}^{\ast}\right)  =\left(  r-p\right)  -\frac{p(1-p)}{p}=r-1.$
Finally, one checks that $R=R^{\ast},$ $R^{2}=R,$ and $\operatorname*{rank}R=$
$\operatorname*{tr}\left(  Q+\frac{1}{\left\Vert \mathbf{x}\right\Vert ^{2}%
}\mathbf{xx}^{\ast}\right)  =r-1+1=r.\hfill$
\end{proof}

\begin{theorem}\label{Proj}\label{thm:projection}
Let \(C=[c_{ij}]\in M_n\) be Hermitian and suppose that \(r\in\{1,\ldots,n\}\).
If all principal submatrices of \(C\) of order \(n-r+1\) are positive
semidefinite, then \(C\circ P\succeq0\) for every rank-\(r\) orthogonal
projection \(P\in M_n\).
Conversely, if \(C\circ P\succeq0\) for every such \(P\), then
every principal submatrix of \(C\) of order \(n-r+1\) is positive semidefinite.
\end{theorem}

\begin{proof}
%\rev{Every smaller principal submatrix is positive semidefinite as well, by extending its index set to one of order $n-r+1$.} 
Since $n-r+1\geq1$ and each diagonal entry of $C$ is a principal submatrix of
order one, it follows that all $c_{ii}\geq0.$ If $r=n,$ then $P=I$ and $C\circ
I\succeq0.$ If $r=1,$ then $C$ and $P$ are positive semidefinite and hence
$C\circ P\succeq0.$ Thus, we need to consider only the cases $n\geq3$ and
$2\leq r\leq n-1.$ We proceed by induction. Let $m\geq2$ and assume that the
theorem has been proved for all pairs $(n,r)$ with $1\leq r\leq n\leq m.$
\newline
\indent
Let $n=m+1,$ suppose that $r\in\{2,\ldots,m\},$ and partition%
\[
C=\left[
\begin{array}
[c]{cc}%
C_{1} & \mathbf{y}\\
\mathbf{y}^{\ast} & c
\end{array}
\right]  ,
\]
in which $C_{1}\in M_{m},$ $\mathbf{y}=[y_{i}]\in\mathbb{C}^{m},$ $c\geq0,$
and every principal submatrix of $C$ (and hence also of $C_{1}$) of order
$m-r+2$ is positive semidefinite. Since $P=[p_{ij}]\in M_{m+1}$ is a rank-$r$
orthogonal projection, Lemma \ref{Lemma} (a) tells us that $p=p_{m+1,m+1}%
\in\lbrack0,1].$
\newline
\indent
If $p=0$ or $p=1,$ then Lemma \ref{Lemma} (b) ensures that $P=Q\oplus\lbrack
p],$ in which $Q$ is an orthogonal projection and $\operatorname*{rank}%
Q=r-p\in\{r,r-1\}.$ The induction hypothesis now ensures that $C_{1}\circ
Q\succeq0$ and hence $C\circ P=\left(  C_{1}\circ Q\right)  \oplus\lbrack
cp]\succeq0.$
\newline
\indent
If $0<p<1,$ then Lemma \ref{Lemma} (c) ensures that%
\[
P=\left[
\begin{array}
[c]{cc}%
Q+\frac{1}{p}\mathbf{xx}^{\ast} & \mathbf{x}\\
\mathbf{x}^{\ast} & p
\end{array}
\right]  ,
\]
in which $Q\in M_{m}$ is an orthogonal projection, $\operatorname*{rank}%
Q=r-1\leq m-1,$ and $Q\mathbf{x}=\mathbf{0}.$ Then%
\begin{equation}
C\circ P=\left[
\begin{array}
[c]{cc}%
C_{1}\circ\left(  Q+\frac{1}{p}\mathbf{xx}^{\ast}\right)   & \mathbf{y}%
\circ\mathbf{x}\\
\left(  \mathbf{y}\circ\mathbf{x}\right)  ^{\ast} & cp
\end{array}
\right]  .\label{CP1}%
\end{equation}
(i) Suppose that $c>0$ in (\ref{CP1}). Then $cp>0$ and $C\circ P\succeq0$ if
and only if the Schur complement $T$ of $cp$ in $C\circ P$ is positive
semidefinite. We have%
\begin{align*}
T  & =C_{1}\circ\left(  Q+\frac{1}{p}\mathbf{xx}^{\ast}\right)  -\frac{1}%
{cp}\left(  \mathbf{y}\circ\mathbf{x}\right)  \left(  \mathbf{y}%
\circ\mathbf{x}\right)  ^{\ast}\\
& =C_{1}\circ Q+\frac{1}{p}C_{1}\circ\mathbf{xx}^{\ast}-\frac{1}{cp}\left(
\mathbf{xx}^{\ast}\circ\mathbf{yy}^{\ast}\right)  \\
& =\left(  1-p\right)  Q\circ C_{1}-\frac{1-p}{c}Q\circ\mathbf{yy}^{\ast
}+\frac{1-p}{p\left(  1-p\right)  }C_{1}\circ\mathbf{xx}^{\ast}\\
& \qquad-\frac{1-p}{cp\left(  1-p\right)  }\left(  \mathbf{xx}^{\ast}%
\circ\mathbf{yy}^{\ast}\right)  +pQ\circ C_{1}+\frac{1-p}{c}Q\circ
\mathbf{yy}^{\ast}\\
& =\left(  1-p\right)  \left(  Q+\frac{1}{p\left(  1-p\right)  }%
\mathbf{xx}^{\ast}\right)  \circ\left(  C_{1}-\frac{1}{c}\mathbf{yy}^{\ast
}\right)  +Q\circ\left(  pC_{1}+\frac{1-p}{c}\mathbf{yy}^{\ast}\right)  \\
& =\left(  1-p\right)  \left(  C_{1}-\frac{1}{c}\mathbf{yy}^{\ast}\right)
\circ\left(  Q+\frac{1}{\left\Vert \mathbf{x}\right\Vert ^{2}}\mathbf{xx}%
^{\ast}\right)  +\left(  pC_{1}+\frac{1-p}{c}\mathbf{yy}^{\ast}\right)  \circ
Q.
\end{align*}
Every principal submatrix of $C_{1}-\frac{1}{c}\mathbf{yy}^{\ast}$ of order
$m-r+1=n-r$ is obtained in the following way: Take a principal submatrix $M$ of
$C$ that has order $m-r+2$ and contains the diagonal entry $c;$ notice that
$M$ is positive semidefinite. Take the Schur complement of $c$ in $M$, which
is positive semidefinite and is a principal submatrix of $C_{1}-\frac{1}%
{c}\mathbf{yy}^{\ast}$ of order $m-r+1=n-r$. We conclude that every principal
submatrix of $C_{1}-\frac{1}{c}\mathbf{yy}^{\ast}$ of order $m-r+1$ is
positive semidefinite. Since $R=Q+\frac{1}{\left\Vert \mathbf{x}\right\Vert
^{2}}\mathbf{xx}^{\ast}$ is a rank-$r$ orthogonal projection (Lemma
\ref{Lemma} (d)), the induction hypothesis ensures that $\left(  C_{1}%
-\frac{1}{c}\mathbf{yy}^{\ast}\right)  \circ R\succeq0.$ Each principal
submatrix of order $m-r+2$ of both $C_{1}$ and $\mathbf{yy}^{\ast}$ is
positive semidefinite, so the same is true of $pC_{1}+\frac{1-p}{c}%
\mathbf{yy}^{\ast}.$ Since $Q$ is a rank-$(r-1)$ orthogonal projection, the
induction hypothesis ensures that $\left(  pC_{1}+\frac{1-p}{c}\mathbf{yy}%
^{\ast}\right)  \circ Q\succeq0.$
\newline
(ii) Suppose that $c=0$ in (\ref{CP1}). Since $r\leq m$ and every principal
submatrix of $C$ of order $m-r+2\geq2$ is positive semidefinite, each
$2\times2$ principal submatrix%
\[
\left[
\begin{array}
[c]{cc}%
c_{ii} & y_{i}\\
\overline{y_{i}} & 0
\end{array}
\right]  ,\quad i\in\{1,\ldots,m\}
\]
of $C$ is positive semidefinite. This is possible only if $y_{i}=0,$ and hence
$\mathbf{y}=\mathbf{0}$. Then $C=C_{1}\oplus\lbrack0]$ and every principal
submatrix of order $m-r+2$ is positive semidefinite. Let $C^{\prime}%
=C_{1}\oplus\lbrack1]$. Every principal submatrix of order $m-r+2$ of
$C^{\prime}$ is positive semidefinite, so the preceding case (i) ensures that
\[
C^{\prime}\circ P=\left[
\begin{array}
[c]{cc}%
C_{1}\circ\left(  Q+\frac{1}{p}\mathbf{xx}^{\ast}\right)   & \mathbf{0}\\
\mathbf{0}^{\ast} & p
\end{array}
\right]  \succeq0.
\]
Consequently, $C_{1}\circ\left(  Q+\frac{1}{p}\mathbf{xx}^{\ast}\right)
\succeq0,$ and hence%
\[
C\circ P=\left[
\begin{array}
[c]{cc}%
C_{1}\circ\left(  Q+\frac{1}{p}\mathbf{xx}^{\ast}\right)   & \mathbf{0}\\
\mathbf{0}^{\ast} & 0
\end{array}
\right]  \succeq0.
\]
	
Conversely, assume positivity for every rank-\(r\) projection. Then, for any set \(S\) of
size \(n-r+1\), simultaneously permute the rows and columns of \(C\) so that \(C[S]\) becomes the leading principal block, and take
\[
P=\frac{1}{n-r+1}J_{n-r+1}\oplus I_{r-1},
\]
where \(J_j\) is the \(j\times j\) all-ones matrix.
Thus, the leading block of \(C\circ P\) is \(C[S]/(n-r+1)\), so positivity
of all such products forces every such \(C[S]\) to be positive semidefinite.
\end{proof}

\begin{proof}[Proof of Theorem~\ref{Big} and Corollary~\ref{cor:condition}]
Every principal submatrix of the Hermitian matrix \(C=A-\mu_{n-r+1}(A)I\)
of order \(n-r+1\) is positive semidefinite, so Theorem~\ref{Proj} ensures
that \(C\circ P\succeq0\), that is,
\begin{equation}\label{Hypothesis1}
A\circ P\succeq\mu_{n-r+1}(A)(I\circ P),
\end{equation}
for every rank-\(r\) orthogonal projection \(P\in M_n\).
Let \(B=\sum_{i=1}^{r_B}\lambda_i\mathbf{u}_i\mathbf{u}_i^*\)
be a spectral decomposition. For each selected rank
\(r\in\{1,\ldots,r_B\}\), set \(P_r=\sum_{i=1}^r\mathbf{u}_i\mathbf{u}_i^*\).
Then
\[
B-\lambda_rP_r
=\sum_{i=1}^r(\lambda_i-\lambda_r)\mathbf{u}_i\mathbf{u}_i^*
+\sum_{i=r+1}^{r_B}\lambda_i\mathbf{u}_i\mathbf{u}_i^*\succeq0.
\]
The Schur Product Theorem and \eqref{Hypothesis1} imply that
\[
A\circ B\succeq\lambda_r(B)(A\circ P_r)
\succeq\mu_{n-r+1}(A)\lambda_r(B)(I\circ P_r).
\]
This proves \eqref{eq:main}. Taking least eigenvalues for each \(r\)
gives \eqref{eq:optimized}.
The case \(r=r_B\) gives \eqref{eq:full-range}.
For the orthogonal projection \(P_{r_B}\) onto the range of \(B\),
spectral decomposition also gives
\[
P_{r_B}=\sum_{i=1}^{r_B}\mathbf{u}_i\mathbf{u}_i^*
\succeq\sum_{i=1}^{r_B}\frac{\lambda_i}{\lambda_1}
\mathbf{u}_i\mathbf{u}_i^*=\frac{1}{\lambda_1(B)}B.
\]
Taking diagonals gives
\(I\circ P_{r_B}\succeq(I\circ B)/\lambda_1(B)\).
Substituting this into \eqref{eq:main} with \(r=r_B\) proves
\eqref{eq:condition-matrix}; taking least eigenvalues then proves
\eqref{eq:condition-scalar}.
\end{proof}

\section{Sharpness and shifted positivity}\label{sec:sharp}
\begin{proposition}\label{prop:optimal}
For fixed \(A\succeq0\) and a given \(r\in\{1,\ldots,n\}\), the quantity \(\mu_{n-r+1}(A)\) is the
largest real number \(t\) such that, for every rank-\(r\) orthogonal
projection \(P\),
\[
A\circ P\succeq t(I\circ P).
\]
Moreover, equality in the associated scalar bound is attained by a projection.
\end{proposition}

\begin{proof}
Validity follows from \eqref{Hypothesis1}. Locally put
\(k=n-r+1\), and choose \(S\) attaining \(\mu_k(A)\).
Simultaneously permute the rows and columns of \(A\) so that \(A[S]\)
is the leading principal block, and set \(P=J_k/k\oplus I_{r-1}\).
Then
\begin{equation}\label{eq:sharp-projection}
A\circ P=\frac1k A[S]\oplus
\diag(a_{k+1,k+1},\ldots,a_{nn}).
\end{equation}
For any \(i\), an order-\(k\) principal submatrix containing \(i\)
shows that \(a_{ii}\ge\mu_k(A)\). Consequently,
\(\lambda_{\min}(A\circ P)=\mu_k(A)/k
=\mu_k(A)\min_iP_{ii}\), including the case \(r=1\), in which the second
block is absent. If \(t>\mu_k(A)\), the first block of
\(A\circ P-t(I\circ P)\) has a negative eigenvalue. Thus no larger
coefficient works for the full class.
\end{proof}

The construction in \eqref{eq:sharp-projection} was suggested by an
anonymous referee. It extends directly to unequal positive eigenvalues.

\begin{proposition}\label{prop:sharp-condition}
Fix \(A\succeq0\), \(r\in\{2,\ldots,n\}\), and any prescribed \(\kappa>1\).
There is a positive semidefinite \(B\) with \(r_B=r\) and
\(\keff(B)=\kappa\) such that the full-range bound
\eqref{eq:full-range} and the optimized bound \eqref{eq:optimized} are both
attained.
\end{proposition}

\begin{proof}
Use the same \(k=n-r+1\), minimizing set \(S\), and coordinate order as
in the proof of Proposition~\ref{prop:optimal}, and take
\[
B=J_k/k\oplus\kappa I_{r-1},\qquad
P_r=J_k/k\oplus I_{r-1}.
\]
The positive eigenvalues of \(B\) are \(1\) and \(\kappa\), the latter
with multiplicity \(r-1\). Hence \(r_B=r\), \(\lambda_r(B)=1\), and
\(\min_i(P_r)_{ii}=\min_i b_{ii}=1/k\). Moreover,
\[
A\circ B=\frac1k A[S]\oplus
\kappa\diag(a_{k+1,k+1},\ldots,a_{nn}).
\]
Since \(a_{ii}\ge\mu_k(A)\ge0\), we have \(\lambda_{\min}(A\circ B)=\mu_k(A)/k\). This is exactly the full-range bound \eqref{eq:full-range},
\[
\mu_k(A)\lambda_{r_B}(B)\min_i(P_{r_B})_{ii}=\mu_k(A)/k.
\]
Thus, the full-range term already equals \(\lambda_{\min}(A\circ B)\),
so the optimized bound \eqref{eq:optimized} is also attained.
\end{proof}

In this construction the simpler bound \eqref{eq:condition-scalar} equals
\(\frac{\mu_{n-r+1}(A)}{\kappa(n-r+1)}\), smaller by the factor \(\kappa\)
when the bound is positive.

To obtain equality with a strictly positive lower bound while keeping both
factors singular, take \(A=I_n-J_n/n\) and \(r\in\{2,\ldots,n-1\}\). Then
\(r_A=n-1\) and \(\mu_{n-r+1}(A)=(r-1)/n>0\), while the
\(B\) constructed in the proof of Proposition~\ref{prop:sharp-condition} has rank \(r<n\). Both the full-range bound \eqref{eq:full-range} and the optimized bound \eqref{eq:optimized} equal
\[
\lambda_{\min}(A\circ B)=\frac{r-1}{n(n-r+1)}>0.
\]
The rank-one case is different: a nonzero rank-one positive semidefinite
matrix necessarily has effective condition number one.
Proposition~\ref{prop:optimal} gives an equality example in this case,
but the equality value is zero if \(A\) is singular.

These constructions do not characterize all pairs \(A,B\) for which equality holds.

\begin{proposition}\label{prop:strict}
Suppose \(A,B\succeq0\), \(b_{ii}>0\) for every \(i\),
and \(\mu_{n-r_B+1}(A)>0\). If \(\keff(B)>1\), then the difference in
\eqref{eq:condition-matrix} is positive definite, and consequently
\eqref{eq:condition-scalar} is strict. 
\end{proposition}

\begin{proof}
Let \(k=n-r_B+1\), and let \(P\) be the
orthogonal projection onto the range of \(B\). Decompose the difference in \eqref{eq:condition-matrix} as
\begin{align*}
&A\circ\bigl(B-\lambda_{r_B}(B)P\bigr)\\
&\quad+\lambda_{r_B}(B)
       \bigl[A\circ P-\mu_k(A)(I\circ P)\bigr]\\
&\quad+\frac{\mu_k(A)}{\keff(B)}
       \bigl[I\circ(\lambda_1(B)P-B)\bigr].
\end{align*}
All three terms are positive semidefinite. We next prove that their sum is positive definite using contradiction.
Suppose that their sum has a null space containing a nonzero \(\mathbf{x}\). Then, the last term
would restrict the support of \(\mathbf{x}\) to zero diagonal positions of
\(\lambda_1(B)P-B\). A positive semidefinite matrix with a zero diagonal
entry has a zero corresponding row and column. Thus
\(B=\lambda_1(B)P\) on that support.
Also \(P_{ii}>0\), since \(b_{ii}>0\), and
\(A\circ P\succeq\mu_k(A)(I\circ P)\succ0\).
The first term would then have quadratic form
\[
\bigl(\lambda_1(B)-\lambda_{r_B}(B)\bigr)\mathbf{x}^*(A\circ P)\mathbf{x}>0,
\]
which is a contradiction. Strictness of the scalar inequality follows by taking
least eigenvalues.
\end{proof}

Proposition \ref{prop:strict} shows that under mild conditions the equality in
\eqref{eq:condition-scalar} is possible only when
\(\keff(B)=1\). This result concerns only the condition-number formulation and does
not contradict Proposition~\ref{prop:sharp-condition}.

\paragraph{Optimality of the condition-number coefficient}
Despite the strictness in Proposition~\ref{prop:strict}, the coefficient
\(1/\kappa\) in \eqref{eq:condition-scalar} cannot be increased if it is
to depend only on \(\kappa=\keff(B)\) and to hold uniformly over all
dimensions and positive semidefinite pairs. Indeed, for any fixed
\(\kappa>1\) and \(n\ge2\), take
\[
A_n=J_n,\qquad B_n=\kappa I_n-\frac{\kappa-1}{n}J_n.
\]
The eigenvalues of \(B_n\) are \(1\) and \(\kappa\), the latter with
multiplicity \(n-1\). Thus \(r_{B_n}=n\), \(\keff(B_n)=\kappa\), and
\(\mu_{n-r_{B_n}+1}(A_n)=\mu_1(A_n)=1\). Moreover,
\(A_n\circ B_n=B_n\), so
\[
\frac{\lambda_{\min}(A_n\circ B_n)}
{\mu_1(A_n)\min_i(B_n)_{ii}}
=\frac{1}{\kappa-(\kappa-1)/n}
\longrightarrow\frac1\kappa\qquad(n\to\infty).
\]
Any larger coefficient at this fixed \(\kappa\) therefore fails for
sufficiently large \(n\). The coefficient in
\eqref{eq:condition-matrix} is consequently optimal in the same
dimension-independent sense. This does not determine optimal coefficients
for fixed dimensions or ranks, or for the restricted class in which both
factors are singular, since \(B_n\) here is positive definite.

Theorem \ref{Proj} considers the Hadamard product of an indefinite Hermitian
matrix $C$ and a rank-$r$ orthogonal projection $P$. It says that if
$\mu_{n-r+1}(C)\geq0,$ then $C\circ P\succeq0.$ With help from Corollary~\ref{cor:condition}, we can generalize this result to the case in which $P$ is merely
positive semidefinite (not necessarily idempotent), keep the assumption that
$C$ is indefinite, and still conclude that $C\circ P\succeq0$. The price to
pay is that we must increase the lower bound assumed for $\mu_{n-r+1}(C)$.

\begin{corollary}
\label{Indefinite}\label{cor:shift}Let $C,B\in M_{n}$ be Hermitian. Suppose that $B$ is
positive semidefinite and $B\ne0.$ If%
\begin{equation}
\mu_{n-r_B+1}(C)\geq-\left(  \kappa_{\mathrm{eff}}(B)-1\right)  \lambda_{\min
}(C),\label{muineqality}%
\end{equation}
then $C\circ B$ is positive semidefinite.
\end{corollary}

\begin{proof}
If $\lambda_{\min}(C)\geq0,$ then $C\succeq0$ and $C\circ B\succeq0$ without
any other assumption. Suppose that $\lambda_{\min}(C)<0$ and let
$A=C-\lambda_{\min}(C)I,$ which is positive semidefinite. Then
(\ref{muineqality}) ensures that
\begin{align*}
\mu_{n-r_B+1}(A)  & =\mu_{n-r_B+1}(C)-\lambda_{\min}(C)\\
& \geq-\left(  \kappa_{\mathrm{eff}}(B)-1\right)  \lambda_{\min}(C)-\lambda_{\min
}(C)\\
& =-\kappa_{\mathrm{eff}}(B)\lambda_{\min}(C),
\end{align*}
that is,%
\begin{equation}
\frac{\mu_{n-r_B+1}(A)}{\kappa_{\mathrm{eff}}(B)}\geq-\lambda_{\min}(C).\label{ABC}%
\end{equation}
Corollary~\ref{cor:condition} and (\ref{ABC}) imply that%
\begin{equation}
A\circ B\succeq\frac{\mu_{n-r_B+1}(A)}{\kappa_{\mathrm{eff}}(B)}\left(  I\circ B\right)
\succeq-\lambda_{\min}(C)\left(  I\circ B\right)  .\label{previous}%
\end{equation}
However,%
\[
A\circ B=\left(  C-\lambda_{\min}(C)I\right)  \circ B=C\circ B-\lambda_{\min
}(C)\left(  I\circ B\right)  ,
\]
which, combined with (\ref{previous}), shows that $C\circ B\succeq0.$
\end{proof}

The shifted assertion and Corollary~\ref{cor:condition} are equivalent as
universal statements. For the reverse implication, given \(A\succeq0\),
put
\[
C=A-\frac{\mu_{n-r_B+1}(A)}{\keff(B)}I.
\]
Then
\[
\mu_{n-r_B+1}(C)+(\keff(B)-1)\lambda_{\min}(C)
=(\keff(B)-1)\lambda_{\min}(A)\ge0.
\]
The shifted assertion gives precisely \eqref{eq:condition-matrix}.

\section{Quantitative comparisons and limitations}\label{sec:examples}
\begin{example}\label{example}\label{ex:styan}
The matrices $A,$ $B,$ and $A\circ B$ in (\ref{StyanExample}) have respective
ranks $2,$ $2,$ and $3;$ their respective Kruskal ranks are $2,$ $1,$ and $3.$
Since $k_{A}+r_B=4>3,$ Theorem \ref{YangTheorem} ensures that $A\circ B$ is
positive definite. Notice that $k_{B}+r_{A}=3,$ which illustrates that the
condition in Theorem \ref{YangTheorem} is sufficient but not necessary.
The smallest eigenvalues of the three \(2\times2\) principal submatrices
are \((3-\sqrt5)/2,(3-\sqrt5)/2,1\). Thus \(\mu_2(A)=(3-\sqrt5)/2\).
The positive eigenvalues of \(B\) are \(2\pm\sqrt2\), and
\(\keff(B)=3+2\sqrt2\). The actual least eigenvalue and the simpler bound
in \eqref{eq:condition-scalar} are
\begin{align*}
\lambda_{\min}(A\circ B)&=\frac{5-\sqrt{17}}2\approx0.43845,\\
\frac{\mu_2(A)}{\keff(B)}\min_i b_{ii}
&=\frac{3-\sqrt5}{2(3+2\sqrt2)}\approx0.065535.
\end{align*}
The orthogonal projection onto the range of \(B\) has diagonal \((1,1/2,1/2)\).
The full-range bound \eqref{eq:full-range}, with \(r_B=2\),
also equals the optimized bound \eqref{eq:optimized} and is
\[
\frac{(3-\sqrt5)(2-\sqrt2)}4\approx0.11188.
\]
The decimal values are approximations. This example gives a positive
certificate even though both orientations of the classical Schur bound are zero.

For singular positive semidefinite matrices $A, B$ satisfying $k_A \geq n-r_B+1$, we can always take $C = A - cI$, $0<c \leq\frac{\mu_{n-r_B+1}(A)}{\kappa_{\mathrm{eff}}(B)}$ such that \eqref{muineqality} is satisfied. For the case in Example \ref{example}, we can take $0<c\leq0.0655$, so $C$ is indefinite Hermitian and $C\circ B\succeq 0$.

\end{example}
\begin{example}\label{ex:indefinite}
Let%
\[
C=\left[
\begin{array}
[c]{ccc}%
8 & 7 & 0\\
7 & 8 & 4\\
0 & 4 & 8
\end{array}
\right]  \quad\text{and}\quad P=\frac{1}{3}\left[
\begin{array}
[c]{ccc}%
2 & -1 & -1\\
-1 & 2 & -1\\
-1 & -1 & 2
\end{array}
\right]  .
\]
The matrix $P$ is a rank-$2$ orthogonal projection and $n-r_{P}+1=2$. The
three 2-by-2 principal submatrices of $C$ are positive definite. The
eigenvalues of $C$ are $8\pm\sqrt{65}$ and $8$, one of which is negative. The
eigenvalues of%
\[
C\circ P=\frac{1}{3}\left[
\begin{array}
[c]{ccc}%
16 & -7 & 0\\
-7 & 16 & -4\\
0 & -4 & 16
\end{array}
\right]
\]
are $\frac{1}{3}\left(  16\pm\sqrt{65}\right)  $ and $\frac{16}{3},$ all of
which are positive.
\end{example}

\begin{example}\label{ex:referee}
The following family was supplied by an anonymous referee.
For \(0<t<\pi/2\), let
\begin{equation}\label{eq:referee-family}
A=I_4-\frac14J_4,\qquad
B_t=\begin{pmatrix}
1&0&1&\cos t\\0&1&0&0\\1&0&1&\cos t\\
\cos t&0&\cos t&1
\end{pmatrix}.
\end{equation}
Both factors are positive semidefinite of rank three.
The positive eigenvalues of \(B_t\), in decreasing order, are
\[
\lambda_1(B_t)=\frac{3+\sqrt{9-8\sin^2t}}2,\quad
\lambda_2(B_t)=1,\quad
\lambda_3(B_t)=\frac{3-\sqrt{9-8\sin^2t}}2.
\]
Since \(B_t\) has unit diagonal,
\[
\lambda_{\min}(A\circ B_t)=1-\frac14\lambda_1(B_t)\longrightarrow\frac14.
\]
Here \(\mu_2(A)=1/2\), and the effective-condition-number bound \eqref{eq:condition-scalar} is
\(\frac{\lambda_3(B_t)}{2\lambda_1(B_t)}\to0\).
The orthogonal projection onto the range of \(B_t\) has minimum
diagonal \(1/2\), so the full-range bound \eqref{eq:full-range},
\(\lambda_3(B_t)/4\), also tends to zero.
Thus neither of these two scalar estimates has a positive universal
relative-accuracy guarantee.

Rank selection helps in this example. We have \(k_A=3\) and
\(\mu_3(A)=1/4\). A unit eigenvector of \(B_t\) for \(\lambda_1(B_t)\)
tends to \((1,0,1,1)^T/\sqrt3\), whereas the second coordinate vector
is an eigenvector for \(\lambda_2(B_t)=1\).
The diagonal of \(P_2\) consequently tends to
\((1/3,1,1/3,1/3)\). The \(r=2\) term in \eqref{eq:optimized}
tends to \(1/12\); the other two terms tend to zero.
The optimized bound \eqref{eq:optimized} therefore tends to \(1/12\), not zero.
This improves the estimate on the family without making it exact.
\end{example}

\begin{example}\label{ex:matching}
Keep \(B_t\) from \eqref{eq:referee-family}, but replace \(A\) by
\begin{equation}\label{eq:rank-two-A}
A=\begin{pmatrix}
1&0&1/\sqrt2&1/\sqrt2\\
0&1&1/\sqrt2&-1/\sqrt2\\
1/\sqrt2&1/\sqrt2&1&0\\
1/\sqrt2&-1/\sqrt2&0&1
\end{pmatrix}.
\end{equation}
Indeed, \(A=X^*X\), where the four columns of
\[
X=\begin{pmatrix}
1&0&1/\sqrt2&1/\sqrt2\\
0&1&1/\sqrt2&-1/\sqrt2
\end{pmatrix}
\]
are unit vectors. Any two columns of \(X\) are independent.
Since \(X^*\) is injective on \(\mathbb R^2\), the same dependence
relations hold among the corresponding columns of \(A=X^*X\). Thus
\[
r_A=k_A=2,\qquad
\mu_2(A)=1-\frac1{\sqrt2},\qquad \mu_3(A)=\mu_4(A)=0.
\]
The terms \(r=1,2\) in \eqref{eq:optimized} are zero.
The product $A\circ B_t$ has eigenvalues
\[
1,\quad 1,\quad 1\pm\sqrt{\frac{1+\cos^2t}{2}},
\]
and hence, as \(t\downarrow0\),
\begin{equation}\label{eq:matching-rate}
\lambda_{\min}(A\circ B_t)
=1-\sqrt{\frac{1+\cos^2t}{2}}\sim\frac{t^2}{4}.
\end{equation}
The orthogonal projection \(P_3\) onto the range of \(B_t\) has minimum diagonal \(1/2\),
and \(\lambda_3(B_t)\sim2t^2/3\).
The ratio of the full-range bound \eqref{eq:full-range}
to the actual eigenvalue satisfies
\begin{equation}\label{eq:relative-matching}
\frac{\mu_2(A)\lambda_3(B_t)\min_i(P_3)_{ii}}
{\lambda_{\min}(A\circ B_t)}
\longrightarrow\frac43\left(1-\frac1{\sqrt2}\right)\approx0.3905.
\end{equation}
For the effective-condition-number bound \eqref{eq:condition-scalar} the corresponding ratio tends to
\(\frac89(1-1/\sqrt2)\approx0.2603\).
In this example, both bounds have the correct quadratic order despite
\(\keff(B_t)\to\infty\), and the optimized bound \eqref{eq:optimized}
equals the full-range bound \eqref{eq:full-range}.
\end{example}

Since the smallest eigenvalue of a Hermitian matrix does not exceed any diagonal entry, and \(a_{ii}\le\lambda_1(A)\), we have
\[
\lambda_{\min}(A\circ B)
\le \min_i a_{ii}b_{ii}
\le \lambda_1(A)\min_i b_{ii}.
\]
Consequently, for families of pairs with a positive right-hand side in
\eqref{eq:condition-scalar}, the effective-condition-number bound, and hence the
optimized bound \eqref{eq:optimized}, captures a uniformly positive fraction of the actual
smallest eigenvalue whenever
\[
\frac{\mu_{n-r_B+1}(A)}{\keff(B)\lambda_1(A)}
\]
is uniformly bounded below by a positive constant. This condition is
sufficient but not necessary, as Example~\ref{ex:matching} shows.

\begin{example}\label{ex:positive-definite}
The optimized bound \eqref{eq:optimized} can be better than
the classical Schur bound \eqref{classical lower bound} if both factors
are positive definite. Let \(A=B=I_3-J_3/4\), in which
\(J_3\) is the \(3\times3\) all-ones matrix. Both factors have eigenvalues
\(1,1,1/4\), and \(\mu_2(A)=1/2\). The leading rank-\(2\) projection is
\(P_2=I_3-J_3/3\).
The maximum in \eqref{eq:optimized} is attained at \(r=2\):
the \(r=1\) term is at most \(1/12\), and the \(r=3\) term is \(3/16\).
Therefore,
\[
\begin{aligned}
\frac12=\lambda_{\min}(A\circ B)
&\ge\mu_2(A)\lambda_2(B)\min_i(P_2)_{ii}\\
&=\frac13>\frac3{16}
=\lambda_{\min}(A)\min_i b_{ii}.
\end{aligned}
\]
\end{example}

\section{Hadamard products in array signal processing}\label{sec:array}
In this section, we revisit the role of Hadamard products in array signal processing, building on \cite{yang2019source,yang2023nonasymptotic}, and discuss the implications of the new results developed in this work. Consider the classical direction-of-arrival (DOA) estimation problem, where the goal is to estimate the directions of multiple emitting sources using a uniform linear array (ULA) of sensors.

Let $N$ denote the number of sensors, spaced at half-wavelength intervals, and let $K < N$ be the number of sources. Denote the distinct DOAs by $\{\theta_k \in [-\frac{\pi}{2}, \frac{\pi}{2})\}_{k=1}^K$. Under the standard far-field and narrowband assumptions, the array output at snapshot $l$ is modeled as
\begin{equation}\notag
	\mathbf{y}(l) = V(\mathbf{\omega})\mathbf{s}(l) + \mathbf{e}(l), \quad l\in\{1,\ldots,L\},
\end{equation}
in which $L$ is the number of snapshots, $\mathbf{s}(l) \in \mathbb{C}^K$ is the source signal vector, and $\mathbf{e}(l) \in \mathbb{C}^N$ is additive noise. The steering matrix $V(\omega) \in \mathbb{C}^{N \times K}$ is Vandermonde with entries
\begin{equation}\notag
	[V(\omega)]_{i,k} = e^{\mathrm{i}(i-1)\omega_k}, \quad 
	\omega_k = \pi\sin\theta_k \in \left[-\pi, \pi\right).
\end{equation}
Since $\omega_k$ is in one-to-one correspondence with $\theta_k$, it suffices to estimate $\{\omega_k\}$.

Subspace-based methods such as MUSIC \cite{schmidt1981signal} and ESPRIT \cite{roy1986esprit} exploit the covariance structure of $\mathbf{y}(l)$. Assume
\begin{equation}\notag
	\mathbb{E}[\mathbf{s}(l)\mathbf{s}^*(l')] = \delta_{l,l'}\Sigma,\quad
	\mathbb{E}[\mathbf{e}(l)\mathbf{e}^*(l')] = \delta_{l,l'}\sigma^2 I,\quad
	\mathbb{E}[\mathbf{s}(l)\mathbf{e}^*(l')] = 0,
\end{equation}
in which $\Sigma \in \mathbb{C}^{K \times K}$ is the positive semidefinite source covariance matrix, $\delta_{l,l'}$ is the Kronecker delta, and $\sigma^2\geq0$ denotes the noise power.
We assume zero means and positive source powers
$\Sigma_{ii}>0$ for every $i\in\{1,\ldots,K\}$.
Then the covariance matrix is
\begin{equation}\notag
	\Sigma_y = \mathbb{E}[\mathbf{y}(l)\mathbf{y}^*(l)] = V(\omega)\Sigma V^*(\omega) + \sigma^2 I.
\end{equation}
Its spectral decomposition is
\begin{equation}\notag
	\Sigma_y = \sum_{j=1}^N \lambda_j \mathbf{u}_j \mathbf{u}_j^*,\quad \lambda_1 \ge \cdots \ge \lambda_N.
\end{equation}

If $\Sigma \succ 0$, then
\begin{equation}\notag
	\lambda_1 \ge \cdots \ge \lambda_K > \lambda_{K+1} = \cdots = \lambda_N = \sigma^2,
\end{equation}
so that $\mathbb{C}^N$ decomposes into the signal subspace $\mathcal{S} = \mathrm{span}\{\mathbf{u}_1,\dots,\mathbf{u}_K\}$ and the noise subspace $\mathcal{N} = \mathrm{span}\{\mathbf{u}_{K+1},\dots,\mathbf{u}_N\}$. Moreover,
\begin{equation}\notag
	\mathcal{S} = \mathrm{span}\, V(\omega),
\end{equation}
which enables recovery of $\{\omega_k\}$ via MUSIC or ESPRIT. We omit the details.

The preceding framework relies on two idealized assumptions: (i) $\Sigma \succ 0$, and (ii) $L \to \infty$, so that $\Sigma_y$ is perfectly known. Both assumptions are often violated in practice.

\subsection{Spatial smoothing and Hadamard structure}

When $\Sigma$ is singular (e.g., due to multipath or linearly dependent sources), we have $\lambda_K = \lambda_{K+1}$, and the signal subspace becomes ill-defined. A standard remedy is spatial smoothing, which partitions the array into $P$ overlapping subarrays, with each of size $N-P+1$, and averages their covariance matrices \cite{shan1985spatial}. We require $N-P+1>K$ for the MUSIC/ESPRIT subspace argument. The resulting smoothed covariance matrix is
\begin{equation}\notag
	\widetilde{\Sigma}_y = \frac{1}{P} V_{N-P+1}(\omega)\widetilde{\Sigma} V_{N-P+1}^*(\omega) + \sigma^2 I,
\end{equation}
in which $V_{N-P+1}(\omega) \in \mathbb{C}^{(N-P+1)\times K}$ is a submatrix of $V(\omega)$ composed of its first $N-P+1$ rows and
\begin{equation}\notag
	\widetilde{\Sigma} = \sum_{p=1}^P D^{p-1} \Sigma D^{1-p}, \quad 
	D = \mathrm{diag}(e^{\mathrm{i}\omega_1},\dots,e^{\mathrm{i}\omega_K}).
\end{equation}

A key observation in \cite{yang2019source} is that $\widetilde{\Sigma}$ admits the Hadamard factorization
\begin{equation}\label{eq:hadamard}
	\widetilde{\Sigma} = \Sigma \circ \overline{V_P^*(\omega)V_P(\omega)}.
\end{equation}
For distinct $\{\omega_k\}$, it can be verified that
\begin{equation}\notag
r_{V_P^*(\omega)V_P(\omega)} = k_{V_P^*(\omega)V_P(\omega)} = r_{V_P(\omega)} =\min\{P,K\}.
\end{equation}
By the Schur Product Theorem, $\widetilde{\Sigma} \succ 0$ if $V_P^*(\omega)V_P(\omega) \succ 0$, i.e., $P \ge K$ (thus $N \ge 2K$). Using Theorem~\ref{YangTheorem}, this condition improves to
\begin{equation}\notag
k_{V_P^*(\omega)V_P(\omega)}\geq K - r_{\Sigma} + 1 	\quad \Longleftrightarrow \quad 	P \ge K - r_{\Sigma} + 1, 
\end{equation}
Together with $N-P+1>K$, a choice of $P$ exists exactly when
$N\ge2K-r_{\Sigma}+1$. 
These results recover those in \cite{shan1985spatial,bresler1986exact}, but with much simpler derivations due to the new Hadamard perspective established in \eqref{eq:hadamard}.

\subsection{Finite-sample analysis}

In practice, $\Sigma_y$ and $\widetilde{\Sigma}_y$ are estimated from finite snapshots. Under Gaussian assumptions, \cite{yang2023nonasymptotic} shows that the absolute frequency estimation error is bounded from above by
\begin{equation}\label{errorbound}
\frac{C(K,N,P)}{\sqrt{L}} \cdot
\frac{\max\{\sigma\,\sigma_1(V_{N-P+1}(\omega))\lambda_1^{1/2}(\Sigma),\,\sigma^2\}}
{\sigma_K^3(V_{N-P+1}(\omega))\,\lambda_{\min}(\widetilde{\Sigma})}
\end{equation}
with high probability, in which $C(K,N,P)$ is a constant depending only on $K,N,P$ and $\sigma_j(\cdot)$ denotes the $j$th greatest singular value.

Thus, performance critically depends on $\lambda_{\min}(\widetilde{\Sigma})$.
For every \(r\in\{1,\ldots,r_{\Sigma}\}\), let \(P_r\) be a leading rank-\(r\)
spectral projection of \(\Sigma\). Applying Theorem~\ref{Big} to
\eqref{eq:hadamard} gives
\begin{equation}\label{eq:smoothing-opt}
\lambda_{\min}(\widetilde{\Sigma})\ge
\max_{1\le r\le r_{\Sigma}}
\left\{\mu_{K-r+1}(V_P^*(\omega)V_P(\omega))
\lambda_r(\Sigma)\min_i(P_r)_{ii}\right\}.
\end{equation}
A Hermitian matrix and its entrywise conjugate have the same principal
eigenvalues, which explains the omission of conjugation inside \(\mu_j\).
In particular, Corollary~\ref{cor:condition} yields
\equ{\begin{split}\lambda_{\min}\sbra{\widetilde{\Sigma}}
		&\geq \frac{\mu_{K-r_{\Sigma}+1}(V_P^*(\omega)V_P(\omega))}{\kappa_{\mathrm{eff}}\sbra{\Sigma}}\min_{1\leq i\leq K}\Sigma_{ii} \\
		&= \frac{\widetilde{\sigma}_{K-r_{\Sigma}+1}^2(V_P(\omega))}{\kappa_{\mathrm{eff}}\sbra{\Sigma}}\min_{1\leq i\leq K}\Sigma_{ii},\end{split}\label{bound2} %\label{eq:lambdaminSigmastil}
	}
in which $\widetilde{\sigma}_{K-r_{\Sigma}+1}(V_P(\omega))$ denotes the
smallest of the $(K-r_{\Sigma}+1)$-th greatest singular values of the
$P\times\sbra{K-r_{\Sigma}+1}$ submatrices of $V_P(\omega)$.
The lower bound~(\ref{bound2}) is strictly positive if
$N-K\geq P\geq K-r_{\Sigma}+1$.

\subsection{Implications}

The bound in \eqref{bound2} reveals several key performance drivers.
First, in the small-separation regime of interest, larger angular separations
among sources generally improve the conditioning of submatrices of $V_P(\omega)$ and can thereby
reduce the error bound in \eqref{errorbound}.
Second, weaker source correlations among the real emitting sources
(the direct path) can lead to a smaller
$\kappa_{\mathrm{eff}}(\Sigma)$ and a smaller error bound.
Third, balanced and stronger source powers can also increase
$\lambda_{\min}(\widetilde{\Sigma})$.
These observations are consistent with empirical findings reported
in the literature.

\section{Hadamard products in matrix time series analysis}\label{sec:time}
A fundamental problem in matrix time series analysis is to characterize the temporal dependence structure of matrix-valued observations and to forecast future values. 
To this end, the matrix canonical polyadic (CP) factor model has been proposed in \cite{chang2023modelling,chang2024identification} as an effective framework for capturing the evolution of a matrix-valued process via a low-dimensional latent vector process, while preserving the intrinsic row--column structure of the data. 
Consequently, CP-factor modeling provides a principled dimension-reduction approach and has emerged as an important paradigm in matrix time series analysis \cite{han2024cp,han2022tensor}.

Let $\{Y_t\}$ be a $p \times q$ real-valued matrix time series. The matrix CP-factor model takes the form
\begin{equation}
	\label{Eq_CP}
	Y_t = A X_t B^* + \epsilon_t, \qquad t \ge 1,
\end{equation}
where $X_t = \diag(\boldsymbol{x}_t)$ is a $d \times d$ diagonal matrix with latent process $\boldsymbol{x}_t \in \mathbb{R}^d$, $\epsilon_t$ is a matrix white noise process, and $A \in \mathbb{R}^{p \times d}$ and $B \in \mathbb{R}^{q \times d}$ are the loading matrices. 
The temporal dependence is encoded in $\boldsymbol{x}_t$, while $A$ and $B$ characterize the row and column loading structures, respectively. 
Following \cite{chang2023modelling,chang2024identification}, the columns of $A$ and $B$ are normalized to have unit Euclidean norm.

\subsection{The rank-deficient regime}

Under the full-rank assumption $r_A = r_B = d$, \cite{chang2023modelling} proposed a generalized eigenanalysis-based one-pass procedure for identifying the model order $d$ and estimating the loading matrices $A$ and $B$. 
However, rank-deficient loading matrices arise naturally in CP decompositions, and restricting attention to the full-rank case excludes an important class of models \cite{chang2024identification}. 
Moreover, the method in \cite{chang2023modelling} relies critically on the full-rank assumption and is not applicable when $\min\{r_A,r_B\} < d$. 
To address this limitation, \cite{chang2024identification} developed a unified identification and estimation framework that extends one-pass analysis to the rank-deficient regime.

We therefore consider the general setting where the loading matrices may be rank-deficient. Let
\[
r_A = d_1, \qquad r_B = d_2, \qquad d_1,d_2\in\{1,\ldots,d\}.
\]
To identify $d_1$ and $d_2$ and subsequently recover $A$ and $B$, \cite{chang2024identification} constructs the matrices
\[
M_1 = \sum_{k=1}^K \Sigma_{Y,\xi}(k)\Sigma_{Y,\xi}(k)^*,
\qquad
M_2 = \sum_{k=1}^K \Sigma_{Y,\xi}(k)^* \Sigma_{Y,\xi}(k),
\]
based on the lagged dependence structure of $Y_t$, where $K > 1$ is a prescribed integer and
\begin{equation}
	\label{Eq_Sigma}
	\Sigma_{Y,\xi}(k)
	=
	\frac{1}{n-k}\sum_{t=k+1}^n
	\mathbb{E}\!\left[\{Y_t-\mathbb{E}(\bar Y)\}\{\xi_{t-k}-\mathbb{E}(\bar \xi)\}\right].
\end{equation}
Here, $\bar Y = \frac{1}{n}\sum_{t=1}^n Y_t$, $\bar \xi = \frac{1}{n}\sum_{t=1}^n \xi_t$, and each scalar $\xi_t$ is a linear combination of entries of $Y_t$. 

A key requirement for the one-pass identification procedure is that
\[
r_{M_1} = d_1, \qquad r_{M_2} = d_2,
\]
which ensures correct recovery of the signal subspaces. In addition, the nonzero eigenvalues of $M_1$ and $M_2$ must be uniformly bounded away from zero, as assumed in \cite[Condition 3]{chang2024identification}, to guarantee numerical stability.

\subsection{Hadamard structure and eigenvalue lower bounds}

By symmetry, it suffices to analyze $M_1$. Under model \eqref{Eq_CP} and the error-orthogonality condition in \cite[Condition 1]{chang2024identification}, we have
\equ{
M_1 = A \left(\sum_{k=1}^K G_k B^* B G_k\right) A^* = A \left( G \circ B^* B \right) A^*, \label{Eq_M1}}
where 
\[
G_k = \diag(\mathbf g_k)
=
\frac{1}{n-k}\sum_{t=k+1}^n 
\mathbb{E}\!\left[\{X_t-\mathbb{E}(\bar X)\}\{\xi_{t-k}-\mathbb{E}(\bar \xi)\}\right]
\]
and $G = \sum_{k=1}^K \mathbf g_k \mathbf g_k^*$. Therefore, the spectral properties of $M_1$ are governed by the Hadamard product $G \circ B^* B$. In the rank-deficient regime, both $G$ and $B^* B$ may be singular positive-semidefinite. The key question is whether their Hadamard product retains sufficient positivity and spectral nondegeneracy. This is precisely the scenario addressed by our Hadamard-product theory.

Since $G$ and $B^*B$ have no zero diagonal entries, under the mild assumption $k_G + d_2 > d$
made in \cite[Proposition 1]{chang2024identification}, it follows from Theorem~\ref{YangTheorem} that $G \circ B^* B$ is positive definite.\footnote{The assumption is relaxed to $\max\lbra{k_G + d_2,\, r_G+k_{B^* B}} >d$ in \cite[Proposition 1]{chang2024identification} by symmetry of the Hadamard product $G \circ B^* B$.} Consequently,
\[
r_{M_1} = r_A = d_1.
\]

Moreover, applying Corollary~\ref{cor:condition} yields
\begin{equation}
	\label{Eq_lambda_M_1}
	\lambda_{\min}\bigl(G \circ B^* B\bigr)
	\ge
	\frac{\mu_{d-d_2+1}(G)}{\kappa_{\mathrm{eff}}(B^* B)},
\end{equation}
by noting that $B^*B$ has a unit diagonal.
Combining \eqref{Eq_lambda_M_1} with \eqref{Eq_M1}, we obtain that the minimum positive eigenvalue of $M_1$ satisfies
\[
\lambda_{d_1}(M_1)
\ge
\sigma_{d_1}^2(A)\,
\lambda_{\min}\bigl(G \circ B^* B\bigr)
\ge
\frac{\sigma_{d_1}^2(A)\,\mu_{d-d_2+1}(G)}
{\kappa_{\mathrm{eff}}(B^* B)}.
\]
Therefore, for a sequence of such models, the positive eigenvalues
of $M_1$ are uniformly bounded away from zero if
\begin{equation}\label{eq:uniform-M1}
\frac{\sigma_{d_1}^2(A)\,\mu_{d-d_2+1}(G)}
{\kappa_{\mathrm{eff}}(B^*B)}\ge c>0,
\end{equation}
with $c$ independent of the model dimensions and asymptotic index.
For example, \eqref{eq:uniform-M1} is satisfied if $\sigma_{d_1}(A)$ and
$\mu_{d-d_2+1}(G)$ have positive uniform lower bounds and
$\kappa_{\mathrm{eff}}(B^*B)$ has a finite uniform upper bound.

The uniform condition \eqref{eq:uniform-M1} provides a sufficient criterion for
the eigenvalue requirement on $M_1$ in
\cite[Condition 3(ii)]{chang2024identification}, expressed in terms
of the loading matrices and the lagged dependence structure encoded
in $G$. The corresponding criterion for $M_2$ follows by
interchanging $A,B$ and $d_1,d_2$.

\section{Conclusion}\label{sec:conclusion}
In this paper, a positive lower bound is provided for the eigenvalues of a Hadamard product in which both matrix factors are positive semidefinite and both factors can be singular.
A further result permits one matrix factor to be indefinite Hermitian. Applications in array signal processing and matrix time series analysis are discussed.

Recent advances demonstrate that the Hadamard product also plays an important role in artificial intelligence (AI) \cite{chrysos2025hadamard} and other data analysis tasks \cite{ciaperoni2024hadamard}. It is therefore of interest to explore potential applications of the present work in these domains.

%\paragraph{Acknowledgments}
%We thank the anonymous referees for their careful comments, which have improved the quality of this paper.

%\paragraph{Declaration of AI assistance}
%The authors acknowledge the use of generative AI solely for language editing
%and stylistic refinement in the preparation of the original submission.
%During the revision, generative AI was additionally used to assist with
%mathematical exploration and checking calculations.
%The authors retain full responsibility for the final content of the manuscript.

\end{document}